\documentclass{article}
\usepackage{ijcai23}

\usepackage{times}
\usepackage{soul}
\usepackage{url}
\usepackage[hidelinks]{hyperref}
\usepackage[utf8]{inputenc}
\usepackage[small]{caption}
\usepackage{graphicx}
\usepackage{amsmath}
\usepackage{amsthm}
\usepackage{booktabs}
\usepackage{algorithm}
\usepackage{algorithmic}
\usepackage[switch]{lineno}

\usepackage{multirow}
\usepackage{enumerate}
\newtheorem{definition}{Definition}
\newtheorem{theorem}{Theorem}
\newtheorem{lemma}{Lemma}

\newtheorem{corollary}{Corollary}
\newtheorem{example}{Example}

\usepackage{subcaption}

\title{A Strategy-proof Mechanism For Networked Housing Markets}

\author{
Youjia Zhang
\and
Pingzhong Tang
\affiliations
Institute for Interdisciplinary Information Sciences\\ 
Tsinghua University
\emails
zhangyou19@mails.tsinghua.edu.cn,
kenshinping@gmail.com
}

\begin{document}

\maketitle

\begin{abstract}
This paper studies a house allocation problem in a networked housing market, where agents can invite others to join the system in order to enrich their options. Top Trading Cycle is a well-known matching mechanism that achieves a set of desirable properties in a market without invitations. However, under a tree-structured networked market, existing agents have to strategically propagate the barter market as their invitees may compete in the same house with them. Our impossibility result shows that TTC cannot work properly in a networked housing market. Hence, we characterize the possible competitions between inviters and invitees, which lead agents to fail to refer others truthfully (strategy-proof). We then present a novel mechanism based on TTC, avoiding the aforementioned competition to ensure all agents report preference and propagate the barter market truthfully. Unlike the existing mechanisms, the agents' preferences are less restricted under our mechanism. Furthermore, we show by simulations that our mechanism outperforms the existing matching mechanisms in terms of the number of swaps and agents' satisfaction.
\end{abstract}

\section{Introduction}
\label{sec: introduction}

Market design has been greatly influenced by the theory of house allocation mechanisms that allow agents to express preferences over houses and trade them without monetary compensation. Such a mechanism can be applied to various areas such as kidney exchange \cite{roth2004kidney,sonmez2020incentivized}, house allocation \cite{shapley1974cores,abdulkadirouglu1999house}, and so on. Therefore, it has attracted researchers from various fields, including economics, mathematics, and computer science.

In the groundbreaking paper, \cite{shapley1974cores} first formulated the house allocation problem as a mechanism design problem and developed a well-known matching mechanism, Top Trading Cycle, which is strategy-proof (truthful reporting preference is a dominant strategy) and Pareto efficient (resources are allocated to the maximum level of efficiency). However, they considered the case in which all agents in the housing market are only invited by the organizer. With the significant improvement in communication tools, people are interacting with others more frequently and easily than ever before. It is natural to develop such a mechanism over social networks. Indeed, agents might be interested in inviting their friends to the housing markets in order to enrich their options.

The work of mechanism design over social networks has been initiated by \cite{li2017mechanism}. They revealed that increasing the number of participants can improve the revenue of an auction, which is consistent with the result of \cite{bulow1996auctions}. Taking social networks into consideration in the mechanism design problem is promising and has been developed in various fields such as resource allocation \cite{li2017mechanism}, task collaboration \cite{golle2001incentives}, matching \cite{kawasaki2021mechanism}.

An important open question in matching over a networked housing market is how to develop a mechanism that ensures agents report their information truthfully. For example, an agent might not invite his friends because they would compete for a house he prefers, which reduces other agents' options. Such an issue was first discovered by \cite{kawasaki2021mechanism}. They restricted the preference domain and found that TTC simultaneously satisfies strategy-proof and Pareto efficient under such settings; otherwise, it fails to achieve a set of properties. However, restricting the preference domains contradicts the purpose of matching mechanisms over a networked housing market, which is to enrich agents' options in order to obtain a better allocation.

This paper proposes a novel matching mechanism that ensures strategy-proof without sacrificing all agents' preference domains. Indeed, we reveal the possible competition between inviters and invitees, which leads agents to benefit from misreporting. Inspired by the success of Top Trading Cycle \cite{shapley1974cores} in traditional housing markets, we develop a matching mechanism based on TTC for networked housing markets, called Top Trading Cycle with Diffusion (TTCD). Aside from being strategy-proof, the allocation of TTCD is also stable, such that agents cannot improve from coalitions with their ancestors and descendants. We further show that TTCD has a promising number of swaps, which is preferable for organizers who charge for a swap.

The remainder of the paper is organized as follows. Section \ref{sec: literature} reviews the relevant literature. Section \ref{sec: model and preliminaries} describes the model and a set of desirable properties. In Section \ref{sec: mechanisms and impossibility}, we briefly review the existing mechanisms and discuss the impossibilities. Section \ref{sec: our mechanism} proposes a novel matching mechanism and analyzes its properties. Section \ref{sec: empirical} provides the performance of TTCD by simulations. Finally, we conclude with some closing remarks and discuss possible future research in Section \ref{sec: conclusion}.

\section{Literature Review}
\label{sec: literature}

The seminal work \cite{shapley1974cores} introduced house allocation as a mechanism design problem and proposed the Top Trading Cycle (TTC) mechanism as a solution with several desirable properties. Since then, the design of house allocation mechanisms and TTC have received much attention from both researchers and practitioners. \cite{roth1982incentive} proved that TTC is strategy-proof, individually rational, and Pareto efficient. Furthermore, \cite{ma1994strategy} verified the result of \cite{roth1982incentive} and showed that TTC is the only mechanism that satisfies all those properties.

Several variations of TTC have been studied in the literature. For instance, \cite{alcalde2011exchange} generalized the TTC algorithm to the case in which agents are allowed to report indifference in preference. \cite{morrill2015making} characterized the TTC in terms of justness, which allows students with higher priority to veto an objection. \cite{hakimov2018equitable} proposed an Equitable TTC in order to eliminate avoidable justified envy situations. 

Mechanism design over social networks has been well studied in various fields such as marketing and auction. For example, \cite{emek2011mechanisms} proposed a geometrical reward mechanism for marketing in which agents are rewarded for successfully referring others to purchase a product. \cite{li2017mechanism} introduced an auction over social networks that satisfies several important properties. For more details, readers can refer to the work of \cite{zhao2021mechanism} and the references therein. These studies demonstrate the potential of mechanism design with social networks as a valuable direction of research.

The study of mechanism design in networked housing markets was pioneered by \cite{kawasaki2021mechanism}. They revealed that it is impossible for a mechanism to be strategy-proof and Pareto efficient over a networked housing market. As a response, they proposed a modified TTC, which restricts the preference domain of all agents. \cite{gourves2017object,zheng2020barter} studied the networked housing market where agents are only allowed to trade with their neighbors. \cite{you2022strategy} modified the algorithm of \cite{abdulkadirouglu1999house} for a house allocation problem with existing tenants and social networks. Later on, \cite{yang2022one} extended the work of \cite{kawasaki2021mechanism} into a graph network and enlarged the preference domain of agents who fail to invite others. 

\section{Model and Preliminaries}
\label{sec: model and preliminaries}

Consider a house allocation problem in a social network that consists of an organizer $o$ and a set of $n$ agents $N=\{1, 2, ..., n\}$. For each agent $i \in N$, he is endowed with a house $h_{i}$, and the set of all houses is denoted as $H=\{h_{1}, h_{2}, ..., h_{n}\}$. Note that the organizer $o$ is not endowed with any house.

For each agent $i \in N$, he has a set of children $r_{i} \subseteq N$. Each agent $i \in N$, he has a strict preference $\succ_{i}$ over houses $H$, where $h_{j} \succ_{i} h_{i}$ represents that agent $i$ prefers house $j$ over house $i$. Therefore, we denote $\theta_{i}=(\succ_{i}, r_{i})$ as the type of agent $i$.

Agents are asked to report their types as part of the mechanism. We denote $\theta^{\prime}_{i}=(\succ^{\prime}_{i}, r^{\prime}_{i})$ as the report type of agent $i$ under the mechanism. Specifically, it is impossible to spread the information of the barter market to a non-existing child. Hence, $r^{\prime}_{i} \subseteq r_{i}$. Let $\theta^{\prime}=(\theta^{\prime}_{1}, \theta^{\prime}_{2}, ..., \theta^{\prime}_{n})=(\theta^{\prime}_{i}, \theta^{\prime}_{-i})$ be the reported types of all agents, and $\theta^{\prime}_{-i}$ is the reported types of all agents excluding agent $i$. We denote $\Theta=(\Pi \times r)$ be the reported type space of all agents, where $\Pi$ is the preference list and $r$ action spaces on reporting children.

For a given report profile $\theta^{\prime}$, we generate a directed graph $G(\theta^{\prime})=(V(\theta^{\prime}), E(\theta^{\prime}))$, where $V(\theta^{\prime}) \subseteq N \cup \{o\}$, and edge $e(i, j) \in E(\theta^{\prime})$ means that agent $i$ invites agent $j$ to join the barter market ($j \in r^{\prime}_{i}$). In particular, an agent can only join the market if all his ancestors are in the market and decide to invite their children. Without loss of generality, we assume that the organizer $o$ invites all his children $r_{o}$.

The organizer $o$ aims to design a mechanism with a matching policy that incentivizes agents to invite their children to join the barter market in order to provide a broader range of options for the exchange. A matching policy $x=(x_{i})_{i \in N}$ is a redistribution of houses to the agents, where $x_{i}(\theta^{\prime})\in H$ is the house allocated to agent $i$ under matching $x$. Let $X$ be the set of all possible allocations. 

Given the above settings, the networked housing market is a tuple $(N, \Pi, r)$, and the formal definition of a matching mechanism under a networked market is defined as

\begin{definition}
    The networked matching mechanism $M$ is defined by a matching policy $x: \Theta \to X$.
\end{definition}

\subsection{Properties}
\label{sec: properties}

In this section, we define a set of important properties that a matching mechanism $M$ on the social network should satisfy. All these properties are similar and inspired by related works \cite{kawasaki2021mechanism,yang2022one}.

We begin with the formal definition of individual rationality, which ensures that if all agents report truthfully, they have no loss from joining the barter market.

\begin{definition}[Individual Rationality]
    The networked matching mechanism $M$ is individually rational (IR) if $x_{i}(\theta_{i}, \theta_{-i}) \succeq_{i} h_{i}$ for all agents $i \in N$.
\end{definition}

Strategy-proof is also a desirable property for the matching mechanism, which guarantees that reporting both children and preferences truthfully is a dominant strategy for all agents.

\begin{definition}[Strategy-proof]
    The networked matching mechanism $M$ is strategy-proof (SP) if $x_{i}(\theta_{i}, \theta^{\prime}_{-i}) \succeq_{i} x_{i}(\theta_{i}^{\prime}, \theta^{\prime}_{-i})$ for all agents $i \in N$.
\end{definition}

A Pareto efficient mechanism provides an outcome such that there is no other allocation where an agent can be better off without worsening other agents.

\begin{definition}[Pareto Efficient]
    An allocation $\mu$ Pareto dominates another feasible allocation $\nu$ if the following conditions hold
    \begin{itemize}
        \item $\mu_{i} \succeq_{i} \nu_{i}$ for all $i \in N$,
        \item $\mu_{j} \succ_{j} \nu_{j}$ for some $j \in N$. 
    \end{itemize}
    The networked matching mechanism $M$ is Pareto Efficient (PE) if there are no other feasible allocations that Pareto dominates $x(\theta)$.
\end{definition}

The core property is widely used as a stability concept in cooperative game theory. The followings are the standard notions of core from the matching literature \cite{pycia2012stability}.

\begin{definition}[Blocking Coalition]
\label{def: block coalition}
    Given an allocation $x(\theta) \in X$ (with items set $H_{S} \subseteq H$), we say a set of agents $S \subseteq N$ is a blocking coalition for $x(\theta)$ if there exists an allocation $x^{\prime}(\theta) \in X$ such that 
    \begin{itemize}
        \item $x_{i}^{\prime}(\theta) \in H_{S}$ for all $i \in S$,
        \item $x_{i}^{\prime}(\theta) \succeq_{i} x_{i}(\theta)$ for all $i \in S$,
        \item $x_{j}^{\prime}(\theta) \succ_{j} x_{j}(\theta)$ for some $j \in S$.
    \end{itemize}
\end{definition}

Intuitively, agents in $S$ reallocate the house among themselves to have better allocations. Therefore, if there is no blocking coalition for an allocation, such an allocation is \textbf{stable} and belongs to the core.

\begin{definition}[Core]
    An allocation $x(\theta)$ is in the core if there exists no blocking coalition for it. 
\end{definition}

\begin{lemma}
\label{lemma: pe}
    If an allocation $x(\theta)$ is in the core, then it is also PE.
\end{lemma}

\begin{proof}
    Assume if $x(\theta)$ is not PE, then there exists other feasible allocation $y(\theta)$ which is PE, and a subset of $S$ including all agents blocks $x(\theta)$ with $y(\theta)$, which contradicts the definition of core.
\end{proof}

\section{Existing Mechanisms and Impossibility Results}
\label{sec: mechanisms and impossibility}

So far, we have defined the set of desirable properties that a matching mechanism should satisfy. In this section, we briefly review the existing matching mechanisms over networked housing markets. 

Moreover, we demonstrate that it is not possible for a matching mechanism over a networked housing market to simultaneously achieve IR, SP, and PE without restrictions on agents' preferences. Additionally, we also characterize the competition between inviters and invitees, leading to a mechanism that fails to satisfy SP.

Before introducing the matching mechanisms in detail, the following are some fundamental definitions and notations:

\begin{figure}[htp!]
    \centering
    \includegraphics[width=0.4\linewidth]{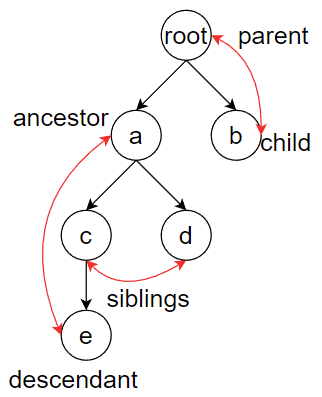}
    \caption{Basic notations}
    \label{fig: notation}
\end{figure}

\begin{itemize}
    \item A directed edge points from a parent node to a child node. (e.g., $a$ is the parent of $c$, and $c$ is the child of $a$ in Figure \ref{fig: notation}.)
    \item An ancestor (descendant) node of a node is either the parent (child) of the node or the parent (child) of some ancestor (descendant) of the node. (e.g., $a$ is the ancestor of $e$, and $e$ is the descendant of $a$.)
    \item Nodes with the same parent are called siblings. (e.g., $c$ is the sibling of $d$.)
\end{itemize}

\subsection{Top Trading Cycle}
\label{sec: TTC}

Top Trading cycle (TTC) is a well-known algorithm for a house allocation problem, which was first proposed in \cite{shapley1974cores}. 

\begin{definition}[Top Trading Cycle]
    TTC algorithm works as follows
    \begin{enumerate}
        \item each agent points to the most preferred house
        \item there must exist at least one cycle with a minimum length $1$
        \item for each cycle, assign each house to the agent pointing to it and remove the cycle from the market
        \item return to step $1$ until no agents remain in the market.
    \end{enumerate}
\end{definition}

Moreover, without social networks, it satisfies all the properties we mentioned in Section \ref{sec: properties} such as IR, SP, and PE \cite{ma1994strategy}. Nevertheless, it is neither SP nor PE under a networked housing market, which is explained in the first impossibility result. 

\subsection{Modified Top Trading Cycle}
\label{sec: modified TTC}

With the success of TTC under general cases, \cite{kawasaki2021mechanism} extended the algorithm into a networked housing market, which is called modified TTC. 

\begin{definition}[modified TTC]
    The modified TTC works as follows
    \begin{enumerate}
        \item each agent points to the most preferred house owned by his \textbf{parents}, \textbf{himself} or \textbf{descendants}
        \item there must exist at least one cycle with a minimum length $1$
        \item for each cycle, assign each house to the agent pointing to it and remove the cycle from the market
        \item return to step $1$ until no agents remain in the market.
    \end{enumerate}
\end{definition}

Despite the fact that modified TTC achieves IR and SP simultaneously, it restricts the preference of agents. Indeed, agents can only choose houses owned by their parents, descendants, and themselves. Furthermore, in Section \ref{sec: empirical}, we show that there may exist an allocation in which Pareto dominates the allocation of modified TTC.

\subsection{Leave and Share}
\label{sec: leave and share}

Later on, \cite{yang2022one} extended the work of modified TTC into a graph network and enlarged the preference domain of particular agents whose parents are removed from the market.

\begin{definition}[Leave and Share]
    The Leave and Share works as follows
    \begin{enumerate}
        \item find the minimum agent $i$ (distance from agent to the organizer)
        
        \item agent $i$ points to his most preferred house $h_{j}$ (owned by agent $j$ who is agent $i$'s \textbf{parents}, \textbf{children} or \textbf{himself}), then agent $j$ does the same action, iteratively, until a cycle with a minimum length $1$ is formed
        
        \item for each cycle, assign each house to the agent pointing to it and remove the cycle from the market
        
        \item reconnect the remaining agents and return to step 1 until no agents are left in the market.
    \end{enumerate}
\end{definition}

Although Leave and Share works well in a graph network, most agents can only exchange with their neighbors (parents and children). 

\subsection{YRMH-IGYT}
\label{sec: yrmh}

\cite{you2022strategy} studied the house allocation problem in which there exist some initially-vacant houses in the networked market. In other words, the number of houses is greater than the number of agents. Note that such a setting leads the problem less complicated than the traditional housing market. The mechanism is called You Request My House - I Get Your Turn (YRMH-IGYT). In order to keep consistency, we assume there are no vacant houses in the market.

\begin{definition}[YRMH-IGYT]
    YRMH-IGYT works as follows
    \begin{enumerate}
        \item find the minimum agent $i$ (distance from agent to the organizer)

        \item agent $i$ points to his most preferred house $h_{j}$ (owned by agent $j$ who is agent $i$'s \textbf{ancestors}, \textbf{children} or \textbf{himself}), then agent $j$ does the same action, iteratively, until a cycle with a minimum length $1$ is formed

        \item for each cycle, assign each house to the agent pointing to it and remove the cycle from the market
        
        \item reconnect the remaining agents and return to step 1 until no agents are left in the market.

    \end{enumerate}
\end{definition}

Similar to modified TTC, YRMH-IGYT restricts the preferences of agents, and there may exist an allocation that Pareto dominates the allocation of YRMH-IGYT.

\subsection{Impossibility Results}
\label{sec: impossibility}

We propose a novel matching mechanism for a networked housing market in response to the following impossibility results and aforementioned challenges, which are also discussed in \cite{kawasaki2021mechanism,yang2022one}.

\begin{theorem}
\label{thm: impossibility 1}
    For a networked housing market $(N, \Pi, r)$ with $n \geq 3$, no mechanism can achieve IR, SP, and PE simultaneously without restricting the preference domain.
\end{theorem}

Due to space constraints, proofs are given in Appendix.

Since it is impossible for a matching mechanism to be IR, SP, and PE simultaneously over the networked housing market. We introduce a weaker definition of the core by taking the network settings into account.

\begin{definition}[Core for Paths]
    For a networked market with $\theta$, there exists a path $p_{i}$ from the organizer $o$ to agent $i \in N$, denoting the set of all agents in $p_{i}$ as $P_{i}$. Given an allocation $x(\theta) \in X$, for any agent $i \in N$, we define an allocation $x(\theta)$ is in the core for paths if no subset of agents in $P_{i}$ can form a blocking coalition.  
\end{definition}

The definition of Core for Paths (CP) is similar to that of the Strict Core for Neighbors (SC4N) in \cite{kawasaki2021mechanism}. However, SC4N restricts the coalitions by two agents in a parent-child relationship, while CP focuses on coalitions formed by agents who are on the same path (agents who share a common ancestor, excluding the organizer).

\begin{example}[CP and SC4N]
\label{eg: matching/C4P SC4N}
    Consider four agents $N=\{s, 1, 2, 3 \}$, where $s$ is the market owner, they have a relationship $r_{s}=\{1 \}$, $r_{1}=\{2 \}$, $r_{2}=\{3 \}$, and $r_{3}=\emptyset$. Consider the following preferences:
    \begin{itemize}
        \item $h_{3} \succ_{1} h_{2} \succ_{1} h_{1}$,
        \item $h_{1} \succ_{2} h_{2} \succ_{2} h_{3}$,
        \item $h_{1} \succ_{3} h_{3} \succ_{3} h_{2}$.
    \end{itemize}
    We have the following two allocations:
    \begin{enumerate}
        \item (\textbf{SC4N}) $x_{1}=h_{1}$, $x_{2}=h_{3}$ and $x_{3}=h_{2}$,
        \item (\textbf{CP}) $x_{1}=h_{2}$, $x_{2}=h_{3}$, and $x_{3}=h_{1}$.
    \end{enumerate}
    
    The allocation ($1$) is SC4P, as agents $1$ and $2$ or agents $2$ and $3$ cannot form a blocking coalition to improve their allocations. However, such an allocation is not CP, as agents $1$, $2$, and $3$ can form a larger coalition group to have a better outcome (allocation ($2$)).
\end{example}

\begin{corollary}
 \label{coro: matching/CP SC4N}
    If an allocation is CP, it is also SC4N; however, if an allocation is SC4N, it may not be CP.
\end{corollary}

The following theorem highlights the key challenge for a matching mechanism over a networked housing market to guarantee agents invite all their children to join the barter market. In the following theorem, we use the term `compete' to refer to the situation where agents $i$ and $j$ both have (point) the same house as their top preference.

\begin{theorem}
\label{thm: impossibility 2}
    For a networked housing market $(N, \Pi, r)$ with $n \geq 3$, a matching mechanism is \textbf{not} SP if it allows agents $i$ and $j \in descendant(i)$ to compete for a house owned by
    \begin{itemize}
        \item an agent $k$ who is an ancestor of agent $i$ (e.g., $k \in ancestor(i)$),
        \item an agent $k$ who is a descendant of agent $j$ (e.g., $k \in descendant(j)$),
        \item an agent $k$ who is both a descendant of agent $i$ and an ancestor of agent $j$ (e.g., $k \in \{descendant(i) \cap ancestor(j)\}$), without restricting the preferences of their ancestors and descendants if they (agents $i$ and $j$) are \textbf{not} selected by an agent between them,
    \end{itemize} 
    where $ancestor(i)$ is the set of agents who are the ancestor of agent $i$, and $descendant(i)$ for the set of descendants of agent $i$.

    Furthermore, if any agent $i$ can select a house owned by agent $j$ with no relationship ($j \notin \{ancestor(i), descendant(i), sibling(i) \}$), such a mechanism is \textbf{not} SP.
    
\end{theorem}

\section{Top Trading Cycle with Diffusion}
\label{sec: our mechanism}

Given the impossibility results in Section \ref{sec: impossibility}, TTC fails to achieve IR, SP, and PE over a networked housing market. Moreover, we also explain how to restrict agents' preferences in order to keep the matching mechanism SP. Therefore, we propose a novel algorithm based on traditional TTC, which is called Top Trading Cycle with Diffusion (TTCD), in order to overcome the aforementioned challenges.

As highlighted by \cite{kawasaki2021mechanism}, the presence of multiple paths to an agent can result in strategic behavior and incompatibility. For example, agents may strategically accept invitations from others. To simplify our analysis, we focus on the social network, which is a directed tree rooted at organizer $o$. (For graph networks, see Appendix.) Furthermore, we allow the organizer to invite multiple agents, which is not well discussed in related literature.

\begin{definition}[Top Trading Cycle with Diffusion]
    TTCD works as follows
    \begin{enumerate}
        \item each agent $i \in r_{o}$ (agents invited by the organizer $o$) points to the most preferred house owned by his \textbf{siblings}, \textbf{himself} or \textbf{descendants}
    
        \item each agent $i \in N \setminus r_{o}$ points to the most preferred house owned by his \textbf{ancestors}, \textbf{himself} or \textbf{descendants}

        \item if agents $i$ and $j \in descendant(i)$ point to the same house owned by \textbf{agent $k \in ancestor(i)$}, update agent $j$ points to his next preferred house

        \item if agents $i$ and $j \in descendant(i)$ point to the same house owned by \textbf{agent $k \in descendant(j)$}, update agent $i$ points to his next preferred house

        \item if agents $i$ and $j \in descendant(i)$ point to the same house owned by \textbf{agent $k \in descendant(i)$ and $k \notin descendant(j)$}, then agent $k$ points to his most preferred house, and such house owner points to his most preferred house iteratively with the following rules until a cycle is formed
        \begin{itemize}
            \item if an agent points to a house owned by \textbf{agent $x \in \{i, ancestor(i)\}$}, \textbf{agent} $x$ points to the most preferred house owned by \textbf{agent} $i$ or $ancestor(i)$
            \item if an agent points to a house owned by \textbf{agent $x \in \{j, descendant(j)\}$}, \textbf{agent} $x$ points to the most preferred house owned by \textbf{agent} $j$ or $descendant(j)$
        \end{itemize}

        \item repeat steps $3$, $4$, and $5$ until there are no conflicts

        \item there must exist at least one cycle with a minimum length $1$

        \item for each cycle, assign each house to the agent pointing to it and remove the cycle from the market
        
        \item return to steps $1$ and $2$ until no agents remain in the market
    \end{enumerate}
\end{definition}

TTCD is easy to understand, and it works similarly to traditional TTC. For agent $i \in r_{o}$ invited by the organizer, they are free to select any house owned by their descendants and siblings. For other agents, if there are no conflicts between agents and their ancestors/descendants, they are free to select any house in the corresponding tree branch. (Recall our analysis focuses on a directed tree network rooted at the organizer $o$.)

Moreover, steps $3$, $4$, and $5$ prevent the conflicts in Theorem \ref{thm: impossibility 2} and guarantee all agents cannot be worse off from inviting others. Compared with the mechanisms in \cite{kawasaki2021mechanism,yang2022one}, the preference restriction in our mechanism is less strict. This makes our mechanism more efficient and flexible than other existing mechanisms.

We demonstrate a running process of TTCD by using the example shown in Figure \ref{fig: ttcd}. The preference list is given in Table \ref{tab: preference}.

\begin{figure}[htp!]
    \subfloat[]{%
        \includegraphics[width=.2\textwidth]{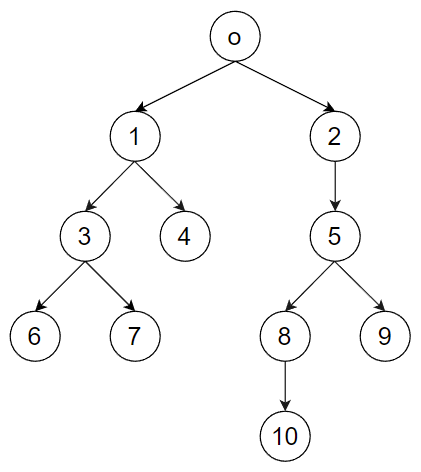}%
        \label{subfig: a}%
    }\hfill
    \subfloat[]{%
        \includegraphics[width=.2\textwidth]{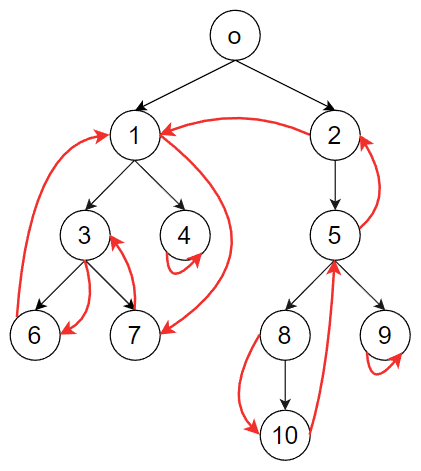}%
        \label{subfig: b}%
    }\\
    \subfloat[]{%
        \includegraphics[width=.2\textwidth]{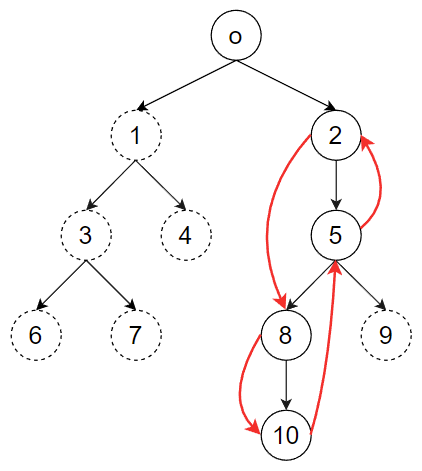}%
        \label{subfig: c}%
    }\hfill
    \caption{A running example of TTCD.}
    \label{fig: ttcd}
\end{figure}

\begin{table}[H]
\centering
\begin{tabular}{|c|c|}
\hline
$i$ & $\succ_{i}$                                                                                \\ \hline
1 & $h_{7} \succ h_{8} \succ h_{1} \succ ...$ \\ \hline
2 & $h_{1} \succ h_{8} \succ h_{7} \succ h_{6} \succ h_{9} \succ h_{5} \succ h_{2} \succ...$ \\ \hline
3 & $h_{6} \succ h_{10} \succ h_{3} \succ ...$ \\ \hline
4 & $h_{9} \succ h_{3} \succ h_{4} \succ...$ \\ \hline
5 & $h_{2}\succ h_{9} \succ h_{8} \succ h_{10} \succ h_{7} \succ h_{5} \succ ...$ \\ \hline
6 & $h_{1} \succ h_{4} \succ h_{3} \succ h_{2} \succ h_{7} \succ h_{10} \succ h_{6} \succ...$ \\ \hline
7 & $h_{4}\succ h_{3} \succ h_{2} \succ h_{10} \succ h_{1} \succ h_{7} \succ ...$ \\ \hline
8 & $h_{10}\succ h_{9} \succ h_{8} \succ ...$ \\ \hline
9 & $h_{3}\succ h_{4} \succ h_{2} \succ h_{7} \succ h_{9} \succ...$ \\ \hline
10 & $h_{5}\succ h_{6} \succ h_{4} \succ h_{8} \succ h_{9} \succ h_{10}\succ...$\\ \hline
\end{tabular}
\caption{Preference list of Figure \ref{fig: ttcd}.}
\label{tab: preference}
\end{table}

\begin{itemize}
    \item Figure \ref{subfig: a} is the directed tree rooted at organizer $o$.
    \item All agents point to their most preferred house available in the market, which is shown in Figure \ref{subfig: b}. Note that agent $9$ cannot point to $h_{2}$, as it is pointed by his ancestor agent $5$. (Recall steps $1$, $2$, $3$, and $4$.)

    \item After the first iteration, agents $1$, $3$, $4$, $6$, $7$, and $9$ are removed from the market.

    \item Agent $2$'s most preferred house in the market is now $h_{8}$.
    
    \item The allocation of agents $N=\{1, 2, 3, 4, 5, 6, 7, 8, 9, 10\}$ is houses $X=\{h_{7}, h_{8}, h_{6}, h_{4}, h_{2}, h_{1}, h_{3}, h_{10}, h_{9}, h_{5}\}$.
\end{itemize}

\subsection{Properties of TTCD}
\label{sec: properties of TTCD}

In this section, we show that TTCD satisfies all the desirable properties we mentioned in Section \ref{sec: properties}.

\begin{lemma}
\label{lemma: ir}
    Top Trading Cycle with Diffusion mechanism satisfies individually rational.
\end{lemma}

IR is obvious, as agent $i$ never points to a house that is worse than $h_{i}$ for him, he is never allocated a house worse than $h_{i}$ under TTCD. Therefore, it is always beneficial for agents to join the system.

\begin{lemma}
\label{lemma: sp}
    Top Trading Cycle with Diffusion mechanism is strategy-proof.
\end{lemma}

\begin{proof}
    Note that the agent $i$'s report type $\theta_{i}^{\prime}$ consists of two information, his preference $\succ_{i}^{\prime}$ and his children set $r_{i}^{\prime}$. 

     (\textbf{true preference $\succ_{i}^{\prime}=\succ_{i}$}) For a fixed $r_{i}^{\prime}$ , we show that agent $i$ cannot obtain a better allocation by reporting $(\succ_{i}^{\prime}, r_{i}^{\prime})$ such that $x((\succ_{i}, r_{i}^{\prime}), \theta_{-i}^{\prime}) \succeq_{i} x((\succ_{i}^{\prime}, r_{i}^{\prime}), \theta_{-i}^{\prime})$.

     \textbf{Case 1}: If agent $i$'s favorite house is the house owned by himself $h_{i}$, he can keep $h_{i}$ immediately by reporting $\succ_{i}$. Moreover, misreporting $\succ_{i}^{\prime}$ leads him to point to a less preferred house and probably form a cycle with other agents. As a result, he is allocated a less preferred house. Hence, it is never optimal for agents to misreport in this case. (This also supplements the proof of IR.)

     \textbf{Case 2}: Assuming agent $i$'s favorite house is $h_{j}$ owned by agent $j$. 

     If there are no conflicts (no agents point to $h_{j}$), under TTCD, agent $i$ always points to $h_{j}$. Hence, the formation of the trading cycle, including agents $i$ and $j$, depends on agent $j$ and other agents, which is irrelevant to $\succ_{i}^{\prime}$. If agent $i$ misreports $\succ_{i}^{\prime}$, it may form a trading cycle with a less preferred house.

    If there exists a conflict (other agents point to $h_{j}$), based on the rule of TTCD, agent $i$'s preference may be restricted, which depends on the conflict type. If agent $i$ is not allowed to point $h_{j}$, which means there exists an agent with a higher priority on selecting $h_{j}$. Thus, whatever agent $i$ misreports $\succ_{i}^{\prime}$, he is never allocated $h_{j}$. 

    If agent $i$ is allowed point $h_{j}$, the formation of the trading cycle depends on agent $j$, which is irrelevant to $\succ_{i}^{\prime}$. The problem goes back to the no conflicts case.

    (\textbf{all children $r_{i}^{\prime}=r_{i}$}). So far, we have proved all agents benefit from reporting true preference. In this part, we need to show that $x((\succ_{i}, r_{i}), \theta_{-i}^{\prime}) \succeq_{i} x((\succ_{i}, r_{i}^{\prime}), \theta_{-i}^{\prime})$, where $r_{i}^{\prime} \subseteq r_{i}$. We only need to consider the situation in that agent $i$ competes with his descendants.

    \textbf{Case 1}: If agent $i$'s favorite house is the house owned by himself $h_{i}$, the allocation of $h_{i}$ is irrelevant to $r_{i}^{\prime}$, and he keeps $h_{i}$ under TTCD. 

    \textbf{Case 2}: Assuming agent $i$'s favorite house is $h_{j}$ owned by agent $j$. 
    
    If there are no conflicts, under TTCD, agent $i$ always points to $h_{j}$. Hence, the formation of the trading cycle, including agents $i$ and $j$, depends on agent $j$ and other agents. If agent $i$ misreports $r_{i}^{\prime}$, it may influence the availability of houses for other agents and fail to form a trading cycle including agents $i$ and $j$.

    If there exists a conflict and agent $i$ is not allowed to point $h_{j}$, whatever he misreports $r_{i}^{\prime}$, he can never form a trading cycle including agent $j$. Moreover, misreporting $r_{i}^{\prime}$ may influence the availability of the next favorite house for agent $i$.
    
    If agent $i$ is allowed point $h_{j}$, the formation of the trading cycle depends on agent $j$. The problem goes back to the no conflicts case.

    Thus, reporting $\theta_{i}^{\prime}=(\succ_{i}, r_{i})$ is optimal under TTCD.

\end{proof}

Lemma \ref{lemma: sp} reveals that reporting private information truthfully is the dominant strategy under TTCD. Indeed, misreporting either the preference $\succ_{i}$ or information of children $r_{i}$ leads to a worse allocation.

According to Theorem \ref{thm: impossibility 1}, as our mechanism is IR and SP, it is \textbf{not} Pareto Efficient. Recall the allocation $X$ in Figure \ref{fig: ttcd}, agents $4$ and $9$ can swap their houses between them to have a better result without affecting other agents' allocations.

\begin{corollary}
\label{coro: not efficient}
    Top Trading Cycle with Diffusion is not Pareto efficient.
\end{corollary}

Although TTCD is not Pareto efficient considering the entire social network, we show that agents cannot collude with their ancestors or descendants to improve their allocations.

\begin{lemma}
\label{lemma: core}
    Over a directed tree networked market, the outcome of the Top Trading Cycle with Diffusion mechanism is in the core for paths.
\end{lemma}
    
Lemma \ref{lemma: core} states that the allocation of TTCD is stable in each path of the tree network such that agents cannot improve their allocations by forming a small coalition group with their ancestors and descendants. Even though there may exist a coalition group in the allocation of TTCD, the agents in the group are not directly connected and thus cannot collude with each other under the network's settings.

As mentioned in \cite{kempe2003maximizing,kawasaki2021mechanism}, the number of swaps is a significant measure for evaluating a matching mechanism for a third-party organizer. Particularly when agents pay the organizer for a swap to a new house. Besides satisfying some desirable properties mentioned in Section \ref{sec: properties}, the organizer aims to maximize the number of swaps as much as possible.

\begin{lemma}
\label{lemma: swap}
    The number of swaps under TTCD is higher than that under the modified TTC \cite{kawasaki2021mechanism}.
\end{lemma}

\section{Empirical Evaluations}
\label{sec: empirical}

In this section, we start with a random example to show the advantages of TTCD compared with other existing mechanisms (modified TTC, Leave and Share, and YRMH-IGYT). Then, we numerically compare these mechanisms by simulations. 

\begin{figure}[H]
    \centering
    \includegraphics[width=0.45\textwidth]{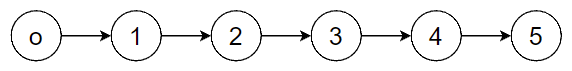}
    \caption{\textbf{Preference}: $h_{5} \succ_{1} h_{4} \succ_{1} h_{2} \succ_{1} h_{1}$, $h_{3} \succ_{2} h_{2}$, $h_{2} \succ_{3} h_{3}$, $h_{1} \succ_{4} h_{5} \succ_{4} h_{4}$ and $h_{1} \succ_{5} h_{4} \succ_{5} h_{5}$.}
    \label{fig: compare}
\end{figure}

Considering the social network in Figure \ref{fig: compare}. The following are the allocations for each mechanism.

\begin{itemize}
    \item \textbf{modified TTC}: $x_{1}=h_{1}$, $x_{2}=h_{3}$, $x_{3}=h_{2}$, $x_{4}=h_{5}$, and $x_{5}=h_{4}$.
    \item \textbf{Leave and Share (LaS)}: $x_{1}=h_{4}$, $x_{2}=h_{3}$, $x_{3}=h_{2}$, $x_{4}=h_{1}$, and $x_{5}=h_{5}$.
    \item \textbf{YRMH-IGYT}: $x_{1}=h_{4}$, $x_{2}=h_{3}$, $x_{3}=h_{2}$, $x_{4}=h_{1}$, and $x_{5}=h_{5}$.
    \item \textbf{TTCD}: $x_{1}=h_{5}$, $x_{2}=h_{3}$, $x_{3}=h_{2}$, $x_{4}=h_{1}$, and $x_{5}=h_{4}$.
\end{itemize}

Although TTCD does not always promise a Pareto efficient allocation, in Figure \ref{fig: compare}, the allocation of TTCD Pareto dominates the allocations of other existing mechanisms.

\subsection{Simulations}

We evaluate the performance of the mechanism by two criteria, the total number of swaps and the average improvement of each agent. The understanding of the number of swaps is intuitive, it indicates how many agents exchange their houses with others. As we discussed in Section \ref{sec: properties of TTCD}, the organizer may aim to maximize the number of swaps. The second criterion, the average improvement of each agent, reflects how far the allocated house is from the initial house. For instance, agent $1$'s preference is $h_{3} \succ_{1} h_{2} \succ_{1} h_{1}$. If he is allocated $h_{3}$ which is in the $1^{st}$ position of his preference and his initial house $h_{1}$ is in the $3^{rd}$ position, hence, he has made a $3-1=2$ position improvement. It is worth mentioning that all existing mechanisms are IR, and therefore, it is impossible for position improvements to be negative.

Moreover, we analyze the performance of the matching mechanism under the different sizes of tree networks. In order to keep consistency, we generate 50 random networks for each different scale of agents.

\begin{figure}[htp!]
    \centering
    \includegraphics[width=0.45\textwidth]{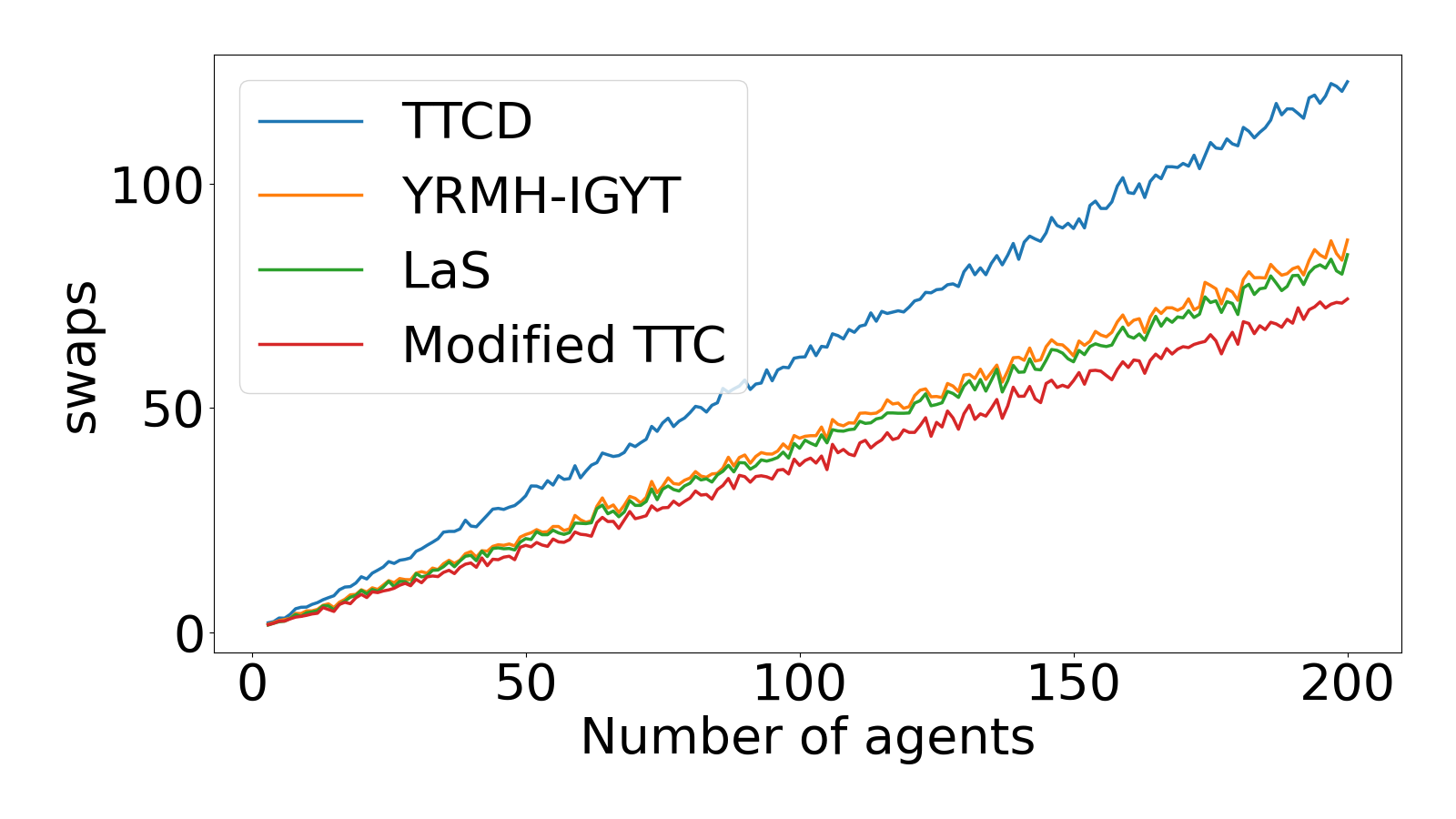}
    \caption{Total number of swaps with different sizes of networks.}
    \label{fig: empirical swaps}
\end{figure}

Figure \ref{fig: empirical swaps} illustrates the performance of four mechanisms in terms of the number of swaps. As modified TTC only allows agents to swap with their parents and descendants, it might be difficult to form a trading cycle with others. As a result, some agents keep their initial houses, leading to a lower number of swaps in the modified TTC. 

Although the restriction of LaS (allowing swaps with parents and children) is stricter than that of modified TTC, it reconstructs the network after removing certain agents, which enlarges some agents' availability. 

Additionally, YRMH-IGYT works similarly to LaS but enlarges the preference domain by allowing agents to select houses owned by their ancestors. Therefore, it generates more swaps than LaS.

In comparison to these mechanisms, TTCD has the least restriction on the preference domain, which results in more swaps, as evident in Figure \ref{fig: empirical swaps}. This observation also supports Lemma \ref{lemma: swap}.

\begin{figure}[htp!]
    \centering
    \includegraphics[width=0.45\textwidth]{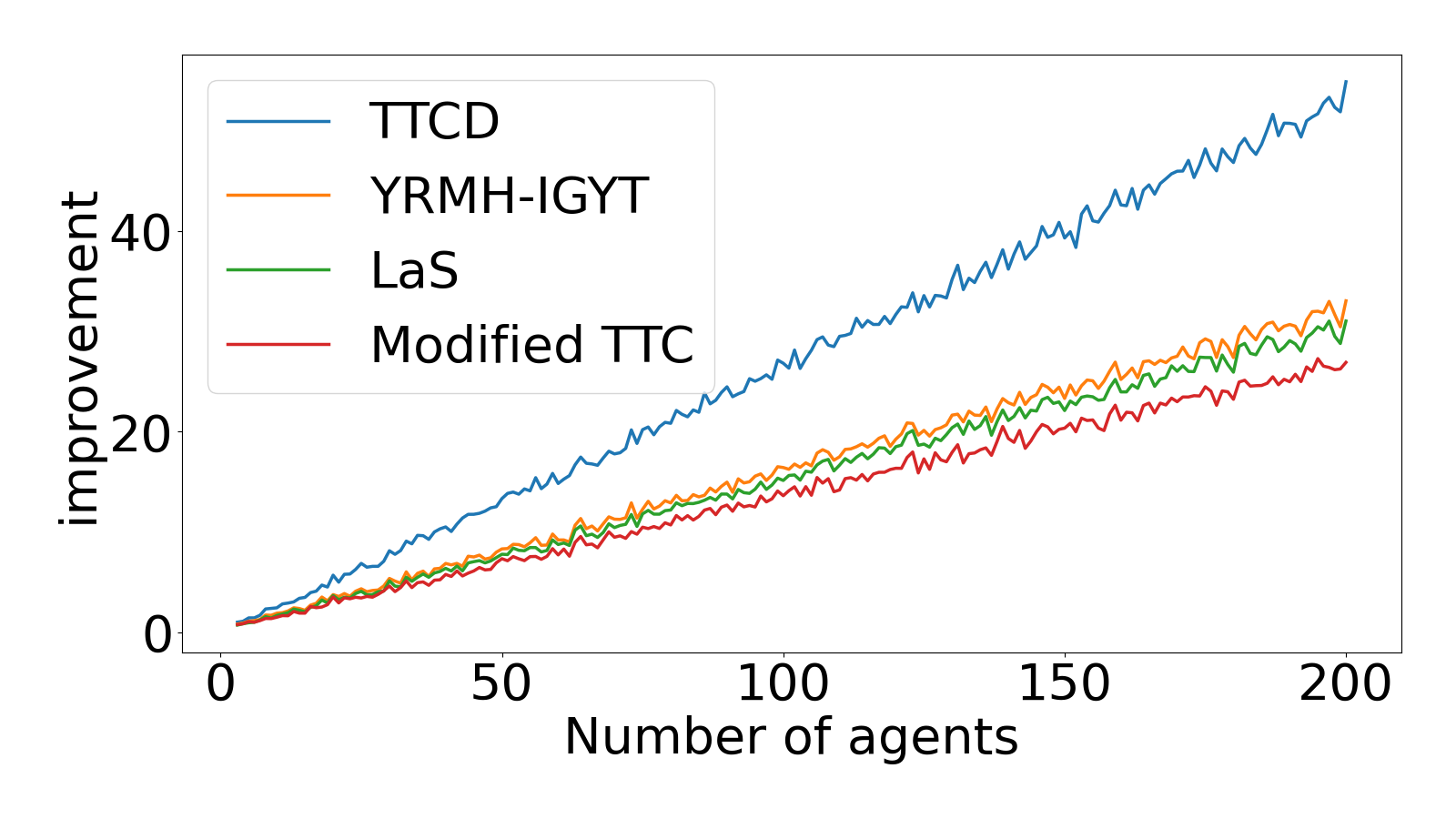}
    \caption{Average number of position improvements for each agent with different sizes of networks.}
    \label{fig: empirical improvement}
\end{figure}

Figure \ref{fig: empirical improvement} shows the improvement in allocation for each agent. As previously stated, the other three mechanisms restrict all agents' preference domains; hence, the probability of forming a large trading cycle is low, and it is impossible for each agent to obtain the most preferred house. Consequently, the position improvement of each agent under these mechanisms is also limited. However, under TTCD, agents are allowed to select houses owned by their ancestors or descendants, which is not possible in the other three mechanisms. Thus, TTCD also outperforms all other mechanisms on position improvements.

\section{Conclusion}
\label{sec: conclusion}

In this paper, we study a matching mechanism over a networked housing market and propose a novel mechanism called Top Trading Cycle with Diffusion (TTCD). In other existing matching mechanisms, they limit all agents' preference domains in order to ensure the truthfulness of the mechanism. We characterize the possible competitions between inviters and invitees, resulting in an untruthful mechanism. Under TTCD, we update the policy based on the traditional TTC in order to avoid all those competitions. As a result, TTCD is strategy-proof which minimizes the restrictions on the preference domain. Besides other desirable properties, it maximizes the agents' satisfaction and the number of swaps.

Promising future work includes considering an allocation problem over social networks with monetary transfers and budget.

\bibliographystyle{named}
\bibliography{ijcai23}

\clearpage

\section{Appendix}

\subsection{TTCD for graph networks}
\begin{definition}[Top Trading Cycle with Diffusion for graph networks]
    TTCD for graph networks works as follows
    \begin{enumerate}
        \item each agent $i \in r_{o}$ (agents invited by the organizer $o$) points to the most preferred house owned by his \textbf{siblings}, \textbf{himself} or \textbf{descendants}
    
        \item each agent $i \in N \setminus r_{o}$ points to the most preferred house owned by his \textbf{ancestors}, \textbf{himself} or \textbf{descendants}

        \item if agents $i$ and $j \in descendant(i)$ point to the same house owned by \textbf{agent $k \in ancestor(i)$}, update agent $j$ points to his next preferred house (note this does \textbf{not} hold if agent $j \in sibling(k)$)

        \item if agents $i$ and $j \in descendant(i)$ point to the same house owned by \textbf{agent $k \in descendant(j)$}, update agent $i$ points to his next preferred house

        \item if agents $i$ and $j \in \{sibling(i) \cap descendant(i)\}$ point to the same house owned by \textbf{agent $k \in \{descendant(i) \cap ancestor(j)\}$}, update agent $j$ points to his next preferred house

        \item if agents $i$ and $j \in descendant(i)$ point to the same house owned by \textbf{agent $k \in descendant(i)$ and $k \notin descendant(j)$}, then agent $k$ points to his most preferred house, and such house owner points to his most preferred house iteratively with the following rules until a cycle is formed
        \begin{itemize}
            \item if an agent points to a house owned by \textbf{agent $x \in \{i, ancestor(i)\}$}, \textbf{agent} $x$ points to the most preferred house owned by \textbf{agent} $i$ or $ancestor(i)$
            \item if an agent points to a house owned by \textbf{agent $x \in \{j, descendant(j)\}$}, \textbf{agent} $x$ points to the most preferred house owned by \textbf{agent} $j$ or $descendant(j)$
        \end{itemize}

        \item repeat steps $3$, $4$, and $5$ until there are no conflicts

        \item there must exist at least one cycle with a minimum length $1$

        \item for each cycle, assign each house to the agent pointing to it and remove the cycle from the market
        
        \item return to steps $1$ and $2$ until no agents remain in the market
    \end{enumerate}
\end{definition}

\subsubsection{TTCD for graph network is strategy-proof.}
\begin{proof}
    We update two rules (steps 3 \& 5) into TTCD. Such rules are used to prevent the following special cases.
    \begin{figure}[htp!]
        \centering
        \includegraphics[width=0.1\textwidth]{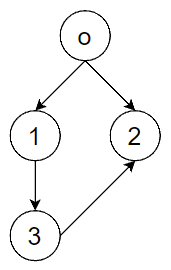}
        \caption{Graph network. Relationship: $r_{o}=\{1, 2\}$, $r_{1}=\{3\}$,$r_{2}=\emptyset$, and $r_{3}=\{2\}$.}
        \label{fig: graph}
    \end{figure}

    \textbf{(Case 1)} Considering the graph network in Figure \ref{fig: graph} with preferences
    \begin{itemize}
        \item $h_{2} \succ_{1} h_{1}$,
        \item $h_{1} \succ_{2} h_{2}$,
        \item $h_{1} \succ_{3} h_{3}$.
    \end{itemize}

    Agent $2$ is the sibling of agent $1$ and also his descendant. If both agents $2$ and $3$ compete for $h_{1}$, and the mechanism restricts the preference of agent $2$ (step 3 in TTCD for tree network), agent $2$ can reject the invitation from agent $3$ or agent $1$ can misreport his children to improve their utilities. Therefore, we update step $3$ (the rule does not hold if agents $1$ and $2$ are siblings).
    
    \textbf{(Case 2)} Considering the graph network in Figure \ref{fig: graph} with preferences
    \begin{itemize}
        \item $h_{3} \succ_{1} h_{2} \succ_{1} h_{1}$,
        \item $h_{3} \succ_{2} h_{1} \succ_{2} h_{2}$,
        \item $h_{2} \succ_{3} h_{3}$.
    \end{itemize}

    If the mechanism allows both agents $1$ and $2$ to compete for $h_{3}$, based on TTCD for tree network, the allocation is $\{h_{1}, h_{3}, h_{2} \}$. However, if agent $1$ misreports agent $3$, resulting in only agents $1$ and $2$ in the network, agents $1$ and $2$ swap their houses. For agent $1$, the allocation is better by misreporting. Therefore, we add one new rule (step 5) to ensure the mechanism is strategy-proof.

    The rest of the proof is the same as that of TTCD for tree networks.
    
\end{proof}

\subsection{Proof of Theorem \ref{thm: impossibility 1}}
\begin{proof}
    \begin{figure}[htp!]
        \centering
        \includegraphics[width=0.3\textwidth]{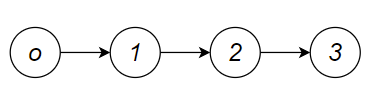}
        \caption{The example to prove Theorem \ref{thm: impossibility 1}. Relationship: $r_{o}=\{1\}$, $r_{1}=\{2\}$, $r_{2}=\{3\}$ and $r_{3}=\emptyset$. Preference: $h_{3} \succ_{1} h_{2} \succ_{1} h_{1}$, $h_{1} \succ_{2} h_{2} \succ_{2} h_{3}$ and $h_{1} \succ_{3} h_{3} \succ_{3} h_{2}$.}
        \label{fig: impossibility 1}
    \end{figure}

    Consider the example given in Figure \ref{fig: impossibility 1}. There are $6$ possible allocations, which are
    \begin{enumerate}
        \item $x_{1} = h_{1}$, $x_{2}=h_{2}$ and $x_{3}=h_{3}$.
        \item $x_{1} = h_{1}$, $x_{2}=h_{3}$ and $x_{3}=h_{2}$.
        \item $x_{1} = h_{2}$, $x_{2}=h_{1}$ and $x_{3}=h_{3}$.
        \item $x_{1} = h_{2}$, $x_{2}=h_{3}$ and $x_{3}=h_{1}$.
        \item $x_{1} = h_{3}$, $x_{2}=h_{1}$ and $x_{3}=h_{2}$.
        \item $x_{1} = h_{3}$, $x_{2}=h_{2}$ and $x_{3}=h_{1}$.
    \end{enumerate}

    Obviously, allocations $(2)$, $(4)$, and $(5)$ fail to achieve IR, as some agents are worse off from joining the barter market. For instance, given the allocation $(2)$, agent $3$ might refuse to join the market and keep $h_{3}$.

    Moreover, allocation $(3)$ Pareto dominates allocation $(1)$, as both agents $1$ and $2$ can have a better result in allocation $(3)$. There exist two Pareto optimal allocations, which are allocations $(3)$ and $(6)$.

    However, given the allocation $(6)$, agent $2$ can have a better allocation by not inviting agent $3$ and forcing agent $1$ to exchange with him, which is allocation $(3)$. As a result, a mechanism that outputs allocation $(6)$ is not SP.

    Now we consider the mechanism outputs the allocation $(3)$. According to the preference list, there always exists a cycle between agents $1$ and $3$. In order to output the allocation $3$, the mechanism has to force agent $3$ pointing other houses rather than $h_{1}$. Therefore, allocation $(3)$ can only be obtained from a non-SP mechanism or restricting the preference domain of particular agents; for example, agent $3$ can only choose $h_{3}$.
    
    More similar results can be found in \cite{kawasaki2021mechanism,yang2022one}. \cite{kawasaki2021mechanism} show that Top Trading Cycle fails to achieve IR, SP, and PE simultaneously in a social network. \cite{yang2022one} prove there is no SP mechanism that outputs a Pareto optimal allocation over a networked housing market.
\end{proof}

\subsection{Proof of Theorem \ref{thm: impossibility 2}}
\begin{proof}
    (\textbf{A house owned by an ancestor.}) Consider the example given in Figure \ref{fig: impossibility 1}. If a mechanism allows both agents $2$ and $3$ to compete for $h_{1}$, agent $2$ may misreport $r_{2}^{\prime}=\emptyset$, so that no one can compete $h_{1}$ with him and he is allocated $h_{1}$ in the end, which is better than $h_{2}$.
    
    Therefore, it is beneficial for agents to misreport if there exists a competition between agents and their descendants in a house owned by their ancestors, which contradicts SP.

    (\textbf{A house owned by a descendant.}) Consider the example given in Figure \ref{fig: impossibility 1} with a preference list $h_{3} \succ_{2} h_{1} \succ_{2} h_{2}$ for agent $2$. If a mechanism allows both agents $1$ and $2$ to compete for $h_{3}$, agent $2$ may misreport $r_{2}=\emptyset$ in order to be allocated $h_{1}$, which is better than $h_{2}$.

    (\textbf{A house owned by an agent between two competitors})

    \begin{figure}[htp!]
        \centering
        \includegraphics[width=0.3\textwidth]{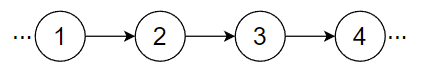}
        \caption{Preference: $h_{2} \succ_{1} h_{1}$, $h_{4} \succ_{2} h_{3} \succ_{2} h_{2}$, $h_{2} \succ_{3} h_{3}$ and $h_{1} \succ_{4} h_{4}$.}
        \label{fig: impossibility between}
    \end{figure}

    Consider the example given in Figure \ref{fig: impossibility between}. Both agents $1$ and $3$ prefer $h_{2}$. If agent $3$ invites agent $4$, $h_{2}$ is allocated to agent $1$. Otherwise, agent $3$ obtains $h_{2}$. As a result, agent $3$ never invites others.
    
    (\textbf{A house owned by an agent in another chain.})
    
    \begin{figure}[htp!]
        \centering
        \includegraphics[width=0.3\textwidth]{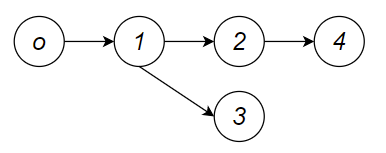}
        \caption{Preference: $h_{4} \succ_{1} h_{3} \succ_{1} h_{2} \succ_{1} h_{1}$, $h_{3} \succ_{2} h_{1} \succ_{2} h_{2} \succ_{2} h_{4}$, $h_{4} \succ_{3} h_{1} \succ_{3} h_{2} \succ_{3} h_{3}$ and $h_{3} \succ_{4} h_{2} \succ_{4} h_{4} \succ_{4} h_{1}$.}
        \label{fig: impossibility 2}
    \end{figure}
    
    Consider the example given in Figure \ref{fig: impossibility 2}. If a mechanism allows both agents $1$ and $3$ to compete for $h_{4}$, agent $1$ can misreport $r_{1}^{\prime}=2$, hence, agent $3$ is not in the market and no one can compete $h_{4}$ with him. 

    Therefore, it is beneficial for agents to misreport if there exists a competition between agents and their descendants in a house owned by an agent in another chain, which contradicts SP.
\end{proof}

\subsection{Proof of Lemma \ref{lemma: core}}
\begin{proof}
    We prove Lemma \ref{lemma: core} by contradiction. Assume there exists a set of agents $S$ in a path $p_{i}$ ($S \subseteq P_{i}$) such that at least one of them has a better allocation without influencing others than under TTCD (e.g. $y_{i}(\theta) \succeq_{i} x_{i}(\theta)$ for all $i \in S$ and $y_{j}(\theta) \succ_{j} x_{j}(\theta)$ for some $j \in S$). 

    We start with the case with $|S|=2$. Note that an agent can only join the system via referrals. Thus, agents in the coalition group $S$ are fully connected to each other. Consider the case that agents $i$ and $j$ (S=\{i,j\}) form a blocking coalition group, and $r_{i}=\{j\}$. By Definition \ref{def: block coalition}, any one of the following holds
    \begin{enumerate}[i.]
        \item $y_{i}(\theta) \succ_{i} x_{i}(\theta)$ and $y_{j}(\theta) \succeq_{j} x_{j}(\theta)$,
        \item $y_{i}(\theta) \succeq_{i} x_{i}(\theta)$ and $y_{j}(\theta) \succ_{j} x_{j}(\theta)$.
    \end{enumerate}

    As both agents $i$ and $j$ form a blocking pair, they are in one trading cycle. Hence, the only solution is to exchange their houses such that $y_{i}=h_{j}$ and $y_{j}=h_{i}$. Furthermore, we can derive the preference list is
    \begin{enumerate}[i.]
        \item $h_{j} \succ_{i} h_{i}$ and $h_{i} \succeq_{j} h_{j}$,
        \item $h_{j} \succeq_{i} h_{i}$ and $h_{i} \succ_{j} h_{j}$.
    \end{enumerate}

    Based on the preference list, under TTCD, agent $i$ also points to $h_{j}$ and agent $j$ points to $h_{i}$ at the same time, the allocation is the same as that in the blocking pair, which contradicts the definition of blocking coalition.

    Due to space constraints, we omit the proof of the case $|S|>2$, which is similar to $|S|=2$. For example, if $S=\{i, j, k\}$ with $r_{i}=j$, $r_{j}=k$, since under TTCD, agents can also point to the house owned by his ancestor, we can consider agents $i$ and $j$ as an agent $i^{\prime}$ with $h_{i^{\prime}}$ and $r_{i^{\prime}}=k$. $h_{i^{\prime}} = h_{i}$ if $h_{i} \succ_{k} h_{j}$; otherwise, $h_{i^{\prime}} = h_{j}$. Then, the problem goes back to $|S|=2$.
    
\end{proof}

\subsection{Proof of Lemma \ref{lemma: swap}}
\begin{proof}
    The main difference between TTCD and modified TTC is the preference domain. Specifically, modified TTC only allows agents to point to a house owned by their parents, descendants, and themselves. However, under TTCD, agents can point to any house owned by their ancestors, descendants, and themselves. Intuitively, the set of ancestors is greater than the set of parents for each agent. Moreover, TTCD also allows agents who are invited by the organizer to select any house owned by their siblings.

    Hence, agents who remain unchanged under modified TTC may be allocated to a better house under TTCD, which increases the number of swaps. 
\end{proof}

\end{document}